\documentclass[conference,letterpaper]{IEEEtran}

\usepackage[utf8]{inputenc} 
\usepackage[T1]{fontenc}
\usepackage{url}
\usepackage{ifthen}
\usepackage{cite}
\usepackage[cmex10]{amsmath} 

\usepackage[table]{xcolor}
\usepackage{amssymb,amsthm}
\usepackage{subcaption}
\usepackage{tikz}
\usepackage{hyperref}
\usepackage[normalem]{ulem}
\usepackage[noend]{algpseudocode}
\usepackage{algorithm}
\usepackage{tabularx,environ}
\usepackage{multirow,array}

\DeclareMathOperator{\wt}{wt}

\DeclareMathOperator{\Span}{span}

\DeclareMathOperator{\ev}{ev}

\DeclareMathOperator{\diag}{diag}

\newtheorem{defn}{Definition}
\newtheorem{thm}{Theorem}

\newtheorem{lm}{Lemma}

\newtheorem{ex}{Example}

\begin{document}
	
	\title{Polar Codes Do Not Have Many Affine Automorphisms} 
	

	\author{%
		\IEEEauthorblockN{Kirill Ivanov and R\"udiger Urbanke}
		\IEEEauthorblockA{\'Ecole Polytechnique F\'ed\'erale de Lausanne (EPFL)\\
			School of Computer and Communication Sciences\\
			CH-1015 Lausanne, Switzerland\\
			Email: \{kirill.ivanov, rudiger.urbanke\}@epfl.ch}
	}
	
	
	\maketitle
	
	\begin{abstract}
		Polar coding solutions demonstrate excellent performance under the list decoding that is challenging to implement in hardware due to the path sorting operations. As a potential solution to this problem, permutation decoding recently became a hot research topic. However, it imposes more constraints on the code structure.
		
		In this paper, we study the structural properties of Arikan's polar codes. It is known that they are invariant under lower-triangular affine permutations among others. However, those permutations are not useful in the context of permutation decoding. We show that, unfortunately, the group of affine automorphisms of Arikan's polar codes asymptotically cannot be much bigger than the group of lower-triangular permutations.
		
	\end{abstract}
	
	
	
	\section{Introduction}
\IEEEPARstart{S}{ince} their invention \cite{Arikan2009polar},  polar codes have become a topic of extensive research regarding their theoretical properties as well as various aspects of their practical implementation. During these years, there has been significant progress in the design of polar-like codes. This has resulted in 3GPP adopting polar codes for the 5G control channel \cite{3gpp.38.212}. However, the standard decoding of polar codes includes list sorting operations that are costly to implement in hardware \cite{Balatsoukas2014llr}. 

Permutation decoding has the potential to offer the solution to both these problems. In short, we take several permutations of the received noisy vector that are in the automorphism group of the code \cite{geiselhart2021autpolar}. This generates different estimates of the information bits. We then process those estimates independently using a simple decoder to produce candidate codewords. The final decoding result is the candidate that is closest to the received vector. This way we can benefit from a higher degree of parallelization and easier implementation. 

There have been many attempts to design codes that perform well under permutation decoding. Notable contributions can be attributed to Kamenev et al. \cite{kamenev2019polar}, Geiselhart et al. \cite{geiselhart2021autpolar} and Pillet et al. \cite{pillet2021polar,pillet2021classification}. However, despite significant progress, the state-of-the-art polar-like solutions such as CRC-aided polar codes \cite{niu2012crc} or polar subcodes \cite{trifonov2016subcodes} remain out of reach in terms of performance. We addressed the question of code design from the theoretical point of view and established the asymptotic inefficiency of permutation decoding of codes invariant under a large number of variable permutations \cite{ivanov2021efficiency}.

In this paper, we study Arikan's classic polar codes, that are constructed so that the successive cancellation (SC) decoding error probability is minimized. Polar codes have a large group of automorphisms \cite{bardet2016algebraic} but all of these permutations correct identical noise patterns under SC decoding \cite{geiselhart2021autrm}. A recent study demonstrates that in principle it is possible that the automorphism group of polar codes is much larger and includes many permutations that are useful in the context of permutation decoding \cite{geiselhart2021autpolar}.

In this paper, we prove that Arikan's polar codes in fact cannot have many useful permutations in their automorphism group. This result is a direct consequence of the recent characterization of the automorphism groups of decreasing monomial codes \cite{li2021complete} which we plug into the framework of partial symmetry  \cite{ivanov2020symmon, ivanov2021efficiency}. Note that the experimental results for short codes that confirm our result can be found in \cite{geiselhart2021autpolar,pillet2021polar}. Our proof implies that constructing the code by minimizing its SC error probability also restricts the number of its affine automorphisms.
	\section{Background}

\subsection{Polar codes}
Let us denote by $[n]$ the set $\{0,\dots,n-1\}$. A $(n=2^m,k)$ polar code \cite{Arikan2009polar} with the set of frozen symbols $\mathcal F$ is a binary linear block code generated by rows with indices $i\in[n]\setminus \mathcal F$ of the matrix $\mathbf A_m=\begin{pmatrix}
	1 & 0\\
	1 & 1
\end{pmatrix}^{\otimes m}$. For a given binary memoryless symmetric (BMS) channel $W$, the set $\mathcal F$ contains the $n-k$ indices with the largest bit error probabilities under successive cancellation decoding.

\subsection{Automorphism groups}
Consider a permutation $\pi$ on $[n]$. Given a codeword $\mathbf c=(c_0,\dots,c_{n-1})$ of $\mathcal C$, applying $\pi$ leads to the vector $\pi(\mathbf c)=(c_{\pi(0)},\dots,c_{\pi(n-1)})$. If the permuted vector $\pi(\mathbf c)$ is also a codeword of $\mathcal C$ for each $\mathbf c\in \mathcal C$, we say that the code $\mathcal C$ is invariant under the action of $\pi$. The set of all such permutations $\pi$ forms a group and this group is called the automorphism group of the code. 

The main objects of interest in this paper are the affine permutations. They can be represented as maps
$$
\mathbf x\to \mathbf A \mathbf x + \mathbf b,\ \mathbf A\in \mathbb F_2^{m\times m}, \mathbf b \in \mathbb F_2^m,
$$
acting on the binary representation of integers from the set $[2^m]$, where the matrix $\mathbf A$ is invertible. A general affine group $GA(m,\mathbb F_2)$ is the group of all affine permutations.

\subsection{Boolean functions and monomial codes}
Let $\{ x_0, \ldots, x_{m-1} \}$ be a collection of $m$ variables taking their values in $\mathbb F_2$, let $\mathbf v = (v_0, \ldots, v_{m-1}) \in \mathbb{F}_2^m$ be any binary $m$-tuple, and let $\wt(\cdot)$ denote the Hamming weight. Then,
\[
x^{\mathbf v} = \prod_{i=0}^{m-1} x_i^{v_i}
\]
denotes a monomial of degree $\wt(\mathbf v)$. 

A function $f(\mathbf x)=f(x_0,\dots,x_{m-1}):\ \mathbb F_2^m\to \mathbb{F}_2$ is called Boolean. Any such function can be uniquely represented as an $m$-variate polynomial:
$$
f(x_0,\dots,x_{m-1})=\sum_{\mathbf v\in \mathbb F_2^m}a_{\mathbf v}x^{\mathbf v},
$$
where $a_{\mathbf v}\in \{0,1\}$. Its evaluation vector $\ev(f(\mathbf x))\in \mathbb F_2^{2^m}$ is obtained by evaluating $f$ at all points $\boldsymbol \alpha_i$ of $\mathbb{F}_2^m$. Note that any length-$2^m$ binary vector $\mathbf c$ can be considered as an evaluation vector of some function $f$. For the rest of the paper, we assume the standard bit ordering of points, i.e., $\boldsymbol\alpha_i$ being the binary expansion of integer $i$.

Consider a binary linear $(n=2^m,k,d)$ code $\mathcal C$ with generator matrix $\mathbf G$. Due to the correspondence between length-$2^m$ binary vectors and $m$-variate polynomials, instead of looking at $\mathbf G$ we can focus on the generating set of the code defined as 
$$
M_{\mathcal C}=\{f_i,0\le i < k | \ev(f_i)=\mathbf G_{i,*}\},
$$
where $\mathbf G_{i,*}$ are the rows of $\mathbf G$.

The generating set of polar codes can be easily deduced from the set of frozen symbols. Rows of matrix $\mathbf A_m$ are evaluation vectors of all possible monomials in $m$ variables. Consequently, any code spanned by a subset of rows of $\mathbf A_m$ has only monomials in its generating set. Such codes are called monomial, and polar codes are a notable example of monomial codes.

\subsection{Derivatives}
The derivative in direction $\mathbf b$ of the Boolean function $f$ is defined as
\begin{equation}
	\label{eq:monderiv}
	(D_{\mathbf b}f)(\mathbf x) = f(\mathbf x+\mathbf b)-f(\mathbf x).
\end{equation}
Given that $\mathbf g$ is the evaluation vector of $f$, from \eqref{eq:monderiv} it follows that  evaluation vector of $D_{\mathbf b}f$ can be computed as  
$$
\ev(D_{\mathbf b}f)=(g_{\mathbf 0}+g_{\mathbf 0\oplus \mathbf b},\dots,g_{\mathbf{n-1}}+g_{\mathbf{(n-1)}\oplus\mathbf b}).
$$

For the monomials, the expression \eqref{eq:monderiv} can be written as
\begin{equation*}
	D_{\mathbf b}x^{\mathbf v} = (\mathbf x+\mathbf b)^{\mathbf v}-x^{\mathbf v}=\prod_{i=0}^{m-1}(x_i+b_i)^{v_i}-\prod_{i=0}^{m-1}x_i^{v_i}.
\end{equation*}
If $\wt(\mathbf b)=1$, i.e., when $\mathbf b$ is an indicator vector $\mathbf e_i$ with the only nonzero entry in position $i$, the directional derivative coincides with the partial derivative $\frac{\partial f}{\partial x_i}$. 

The derivative in direction $\mathbf b$ of code $\mathcal C$ is a binary linear code with generating set
\begin{equation}
	\label{eq:proj}
	M_{\mathcal C\to \mathbf b}=\left\{D_{\mathbf b}f_i |f_i \in M_{\mathcal C}\right\}.
\end{equation}
By definition, $D_{\mathbf b}f_i$ has identical values at coordinates $\mathbf x$ and $\mathbf x+\mathbf b$ for all $\mathbf x\in \mathbb{F}_2^m$, so we can discard the coordinates $\mathbf x+\mathbf b$ and obtain the $(n^{(\mathbf b)}=2^{m-1},k^{(\mathbf b)}=\dim \Span M_{\mathcal C\to \mathbf b},d^{(\mathbf b)})$ code $\mathcal C^{(\mathbf b)}$, that will be further referred to as the derivative code.


\subsection{Successive cancellation decoding}
Assume that the codeword $\mathbf c$ is transmitted through a BMS channel $W$ and the received vector is $\mathbf y$. The successive cancellation algorithm performs bit-by-bit estimation of the vector $\mathbf u$ as
\begin{equation}
	\label{mSCProb}
	\tilde u_i=\begin{cases}\arg\max_{u_i\in \{0,1\}} W^{(i)}(\mathbf y_0^{n-1},\mathbf {\tilde u}_0^{i-1}|u_i), &i\notin\mathcal F,\\
		0&i\in \mathcal F,
	\end{cases}
\end{equation}
where $W^{(i)}=W^{(\{-,+\}^m)}$ and is obtained by the recursive application of channel transformations 
$$W^{(-)}(y_0,y_1|u_0)=\frac{1}{2}\sum_{u_1\in \{0,1\}}W(y_0|u_0\oplus u_1)W(y_1|u_1)$$
and
$$W^{(+)}(y_0,y_1,u_0|u_1)=\frac{1}{2}W(y_0|u_0\oplus u_1)W(y_1|u_1).$$ 
The decoding process can be reformulated as the following recursive procedure:
\begin{enumerate}
	\item Recover $\mathbf c^{(-)}=\mathbf c_0^{n/2-1}\oplus \mathbf c_{n/2}^{n-1}$ from vector $\mathbf y^{(-)}$ that corresponds to the output of  channel $W^{(-)}$.
	\item Recover $\mathbf c^{(+)}=\mathbf c_{n/2}^{n-1}=\mathbf c^{(-)}\oplus \mathbf c_0^{n/2-1}$ from vector $\mathbf y^{(+)}$ that corresponds to the output of channel $W^{(+)}$, assuming that $\mathbf c^{(-)}$ is correct.
	\item Return $\mathbf c=(\mathbf c^{(-)}\oplus \mathbf c^{(+)}|\mathbf c^{(+)})$.
\end{enumerate}
In this perspective, we first (recursively) recover $\mathbf c^{(-)}\in\mathcal C^{(-)}$, assuming the transmission through the synthetic channel $W^{(-)}$, and then use it to (recursively) recover $\mathbf c^{(+)}\in\mathcal C^{(+)}$, assuming the transmission through the synthetic channel $W^{(+)}$. 

Observe now that any codeword of the code $\mathcal C^{(-)}=\{\mathbf c_0^{n/2-1}\oplus \mathbf c_{n/2}^{n-1}|\mathbf c \in \mathcal C\}$ that appears at the first step of the SC recursions can be written as $$\mathbf c^{(-)}=(c_{\mathbf 0}+c_{\mathbf 0\oplus\mathbf e_0},\dots,c_{\mathbf {n/2-1}}+c_{\mathbf{(n/2-1)}\oplus\mathbf e_0}),$$
and therefore the code $\mathcal C^{(-)}$ is a partial derivative of $\mathcal C$ w.r.t. $x_0$. Consequently, at each level of the SC recursions we pick a variable $x_i$ and perform the decomposition of $\mathcal C$ into codes $\mathcal C^{(-)}$ and $\mathcal C^{(+)}$, where the former is a partial derivative and the latter has generating set $M_{\mathcal C^{(+)}}=\{f\in \mathcal C|\frac{\partial f}{\partial x_i}=0\}$. The standard SC decoding implies picking the variable $x_i$ at level $i$, where level 0 is the topmost and at level $m$ we have length-1 codes that correspond to the information and frozen bits.

The essence of permutation decoding is to take several permutations of the received vector $\mathbf y$ so that we get different vectors $\mathbf y{(-)}$ and $\mathbf y^{(+)}$ during the recursions, increasing the chance of successful decoding. All permutations from the automorphism group of the code induce the same codes $\mathcal C^{(-)}$ and $\mathcal C^{(+)}$. When the action of $\pi$ can be expressed as a permutation of the bits of the integers from $[2^m]$, it can be also viewed as a permutation of variables $\{x_0,\dots,x_{m-1}\}$ and its action on the received vector corresponds to the change in the order of partial derivatives in the SC recursions.

\begin{ex}
	Consider the transmission of the all-zero codeword of the $(8,4,4)$ code $\mathcal C$ with $M_{\mathcal C}=\{1,x_0,x_1,x_2\}$ through the erasure channel and assume that the received vector is $\mathbf y=(\epsilon,0,\epsilon,\epsilon,0,\epsilon,0,0)$. Note that any derivative induces a repetition $(4,1,4)$ code $\mathcal C^{(-)}$ with $M_{\mathcal C^{(-)}}=\{1\}$.
	
	If we take $\pi=\{0,1,2\}$, then we get $$\mathbf y^{(-)}=(y_0+y_4,y_1+y_5,y_2+y_6,y_3+y_7)=(\epsilon,\epsilon,\epsilon,\epsilon).$$ This pattern is uncorrectable in $\mathcal C^{(-)}$ and SC decoding fails. On the other hand, $\pi=\{2,1,0\}$ gives $$\mathbf y^{(-)}=(y_0+y_1,y_2+y_3,y_4+y_5,y_6+y_7)=(\epsilon,\epsilon,\epsilon,0),$$ that can be corrected in $\mathcal C^{(-)}$ and SC decoding succeeds.
\end{ex}

	\section{Polar codes cannot have many affine automorphisms}
\label{s:autpolar}
Over the past few years, various researchers have studied the automorphism group of polar codes as well as the construction of codes for permutation decoding. These works mostly focus on the subgroups of $GA(m,\mathbb F_2)$. In this section, we show how the automorphism groups of polar codes fit into our framework and consequently derive that polar codes cannot be invariant under many affine automorphisms.

\subsection{Known automorphisms of polar codes}

\begin{defn}[\cite{bardet2016algebraic}, Definition 3]
	Two monomials of the same degree are ordered as $x_{i_0}\dots x_{i_{w-1}}\preceq x_{j_0}\dots x_{j_{w-1}}$ if and only if for all $q\in [w]$ holds $i_q\le j_q$. This partial order is extended to the monomials of different degrees through divisibility.
\end{defn}

A monomial code $\mathcal C$ is called \textit{decreasing} if for any monomial $x^{\mathbf v}$ from $M_{\mathcal C}$ all monomials $x^{\mathbf t}\preceq x^{\mathbf v}$ are also in $M_{\mathcal C}$. Theorem 2 in \cite{bardet2016algebraic} states that an automorphism group of any decreasing monomial code contains the \textit{lower-triangular affine group} $LTA(m,\mathbb F_2)$, which is a subgroup of $GA(m,\mathbb F_2)$ that includes only matrices $\mathbf A$ that are lower-triangular, and polar codes are decreasing monomial codes \cite[Theorem 1]{bardet2016algebraic}. This is an important result that allows computing the number of minimum-weight codewords of polar codes. However, all permutations from $LTA(m,\mathbb F_2)$ correct identical error patterns in the context of SC decoding and therefore cannot bring any performance improvement \cite[Corollary 2.1]{geiselhart2021autrm}. In other words, the group $LTA(m,\mathbb F_2)$ is \textit{absorbed by the SC decoder}.

The next important step on this road can be attributed to Geiselhart et al., who introduced a larger automorphism group that appears in decreasing monomial codes and is actually useful for the permutation decoding \cite{geiselhart2021autpolar}. Namely, decreasing monomial codes are invariant under the \textit{block lower-triangular affine group} $BLTA(\mathbf s, m)$ for $\mathbf s=(s_0,\dots,s_{l-1}), \sum_{i}s_i=m$, which is the another subgroup of $GA(m,\mathbb F_2)$, where the matrix $\mathbf A$ has form
\begin{equation}
	\label{eq:blta}
	\mathbf A=\begin{pmatrix}
		\mathbf A_{0,0}&  & & \mathbf 0 \\
		\mathbf A_{1,0}& \mathbf A_{1,1} & & \\
		\vdots & \vdots & \ddots&  \\
		\mathbf A_{l-1,0}& \cdots & &\mathbf A_{l-1,l-1}
	\end{pmatrix},
\end{equation}
where $\mathbf A_{i,j}$ are $s_j\times s_i$ submatrices and all entries that lie above the blocks $\mathbf A_{i,i}$ are zero. The length $l$ of the vector $\mathbf s$ as well as its entries can be determined from $M_{\mathcal C}$.

Li et al. later proved that $BLTA(\mathbf s, m)$ is equal to the group of affine automorphisms of decreasing monomial codes \cite{li2021complete}. Pillet et al. also proved that if $s_0>1$, then the group $BLTA((2,1,\dots,1), m)$ is absorbed by the SC decoder \cite{pillet2021classification} (note that $BLTA((2,1,\dots,1), m)\subseteq BLTA(\mathbf s, m)$ whenever $s_0>1$). The size of $BLTA(\mathbf s, m)$ can be computed as 
\begin{equation}
	\label{eq:bltasize}
	|BLTA(\mathbf s, m)|=2^m\prod_{i=0}^{l-1}\left(2^{s_i\gamma_i}\prod_{j=0}^{s_i-1}\left(2^{s_i}-2^j\right)\right),
\end{equation}
where $\gamma_i=\sum_{j<i}s_i$ \cite{geiselhart2021autpolar}. In case of $l=m$ and $\mathbf s=(1,\dots,1)$ it is equal to $|LTA(m)|=2^{\frac{m(m-1)}{2}+m}$ and for $\mathbf s=(m)$ it coincides with $|GA(m,\mathbb F_2)|=O(2^{m^2+m})$.

\subsection{New restrictions on the size of the automorphism group.}
The existing results state that the group of affine automorphisms of polar codes is $BLTA(\mathbf s, m)$, although no constraints on the values $s_i$ are reported. In this section, we demonstrate that the diagonal blocks cannot grow with $m$. More precisely, we prove the following result:

\begin{thm}[Polar codes do not have many affine automorphisms]
	\label{thm:polaraut}
	Consider a BMS channel $W$ and the sequence of polar codes $\{\mathcal C^m\}$ of rate $R=I(W)$ with increasing block lengths $n=2^m$ s.t. the set of frozen symbols of code $\mathcal C^m$ contains $2^m(1-I(W))$ indices that have the largest SC decoding error probabilities. Then codes $\mathcal C^m$ cannot be invariant under $BLTA(\mathbf s, m)$ s.t. there exists a block of size $s_i$ that is an increasing function of $m$.
\end{thm}

To prove Theorem \ref{thm:polaraut}, we first introduce the concept of partial symmetry, then demonstrate the relation between the invariance under $BLTA(\mathbf s, m)$ and the partial symmetry of codes that appear during the recursions of successive cancellation decoding and conclude by showing that the structural properties from polar coding and partial symmetry formulations contradict each other.
\begin{defn}[\cite{ivanov2021efficiency}, Definition 2]
	\label{def:partsym}
	A $(n=2^m,k)$ code $\mathcal C$ is called $t$-symmetric if $t$ of its partial derivatives have equal dimensions and $m-t$ remaining have dimensions that are strictly greater.
\end{defn}
A code is fully symmetric if $t=m$, non-symmetric if $t=1$ and partially symmetric otherwise. 

Let us now show the correspondence between the invariance under $BLTA(\mathbf s, m)$ and the partial symmetry. Consider the matrix $\mathbf A'$ which has a block-permutation-diagonal form $\diag(\mathbf P_0,\dots,\mathbf P_{l-1})$, i.e., all its non-diagonal blocks are zero and all submatrices $\mathbf A'_{i,i}$ are some $s_i\times s_i$ permutation matrices $\mathbf P_i$:
\begin{equation*}
	\mathbf A'=\begin{pmatrix}
		\mathbf P_{0}&  & & \mathbf 0 \\
		& \mathbf P_{1} & & \\
		\vdots & \vdots & \ddots&  \\
		\mathbf 0& \cdots & &\mathbf P_{l-1}
	\end{pmatrix}.
\end{equation*}

Matrix $\mathbf A'$ belongs to $BLTA(\mathbf s, m)$. Observe that it has a block-permutation structure and consequently is an element of the group $\mathcal P_m$ of $m\times m$ permutation matrices. Authors in \cite{geiselhart2021autpolar} and \cite{pillet2021classification} already noted that if a code is invariant under $BLTA(\mathbf s, m)$, it is $s_{l-1}$-symmetric, which corresponds to the topmost step of the SC recursion. In fact, we can prove a much stronger result. 

Let us start with the following lemma:
\begin{lm}
	\label{lm:recauts}
	Consider a monomial code $\mathcal C$ invariant under permutations of form $\diag(\mathbf P_0,\dots,\mathbf P_{l-2},\mathbf P_{l-1})$ and assume that $\mathbf P_{l-1}$ is $s\times s$ matrix. Then the codes $\mathcal C^{(-)}$ and $\mathcal C^{(+)}$ are invariant under permutations of form $\diag(\mathbf P_0,\dots,\mathbf P_{l-2},\mathbf P'_{l-1})$, where $\mathbf P'_{l-1}$ has size $(s-1)\times(s-1)$.
\end{lm}
\begin{proof}
	When the code $\mathcal C$ is monomial, the codes $\mathcal C^{(-)}$ and $\mathcal C^{(+)}$ correspond to the partition of $M_{\mathcal C}$ into two sets of monomials that include and do not include the variable $x_0$, respectively. Both partitions are closed under any automorphism of $\mathcal C$ on variables $\{x_1,\dots,x_{m-1}\}$, and it remains to observe that any such automorphism can be written in matrix form as $\diag(\mathbf P_0,\dots,\mathbf P_{l-2},\mathbf P'_{l-1})$.
\end{proof}

We use this result to prove the general structural property of monomial codes.

\begin{thm}
	\label{thm:bltasym}
	Consider a monomial code $\mathcal C$ that is invariant under $\diag(\mathbf P_0,\dots,\mathbf P_{l-1})$ and define the set of codes that appear at level $\nu_i=\sum_{j>i}s_j$ of the SC recursions
	$$
	\phi(\mathcal C,i)=\{\mathcal C^{(\{-,+\}^{\nu_i})}\}.
	$$
	Then for $i<l$ any element of $\phi(\mathcal C,i)$ is invariant under permutations on variables $\{x_{\nu_i},\dots,x_{\nu_i+s_i-1}\}$.  
\end{thm}
\begin{proof}
	The standard SC schedule implies that the recursion starts from the derivative $\partial/\partial x_0$ and ends with the length-1 codes after taking the derivative $\partial/\partial x_{m-1}$. At level $\nu_i$ we have a code $\mathcal C^{(\{-,+\}^{\nu_i})}$ that is by definition an element of $\phi(\mathcal C,i)$, which is decomposed into codes $\mathcal C^{(-)}$ and $\mathcal C^{(+)}$ w.r.t. variable $x_{\nu_i}$. We apply Lemma \ref{lm:recauts} recursively starting from $\mathcal C$ until we reach the level $\nu_i$ to conclude that any code from $\phi(\mathcal C,i)$ is invariant under the permutations of form $\diag(\mathbf P_0,\dots,\mathbf P_i)$ and consequently the permutations on variables $\{x_{\nu_i},\dots,x_{\nu_i+s_i-1}\}$.
\end{proof}

\begin{ex}
	Consider the $(16,7)$ code $\mathcal C$ with $M_{\mathcal C}=\{1,x_0,x_1,x_2,x_3,x_0x_2,x_0x_3\}$, which is invariant under $BLTA((2,1,1), 4)$. Figure \ref{fig:decomp} demonstrates the generating sets of codes $\mathcal C^{(-)},\mathcal C^{(+)}$ as well as $\mathcal C^{(--)},\mathcal C^{(-+)},\mathcal C^{(+-)}$ and $\mathcal C^{(++)}$ that appear during the SC decoding of $\mathcal C$. Variables used for the decomposition at each level are painted in red. It is easy to verify that all these codes are invariant under the permutations on variables $x_2$ and $x_3$.
	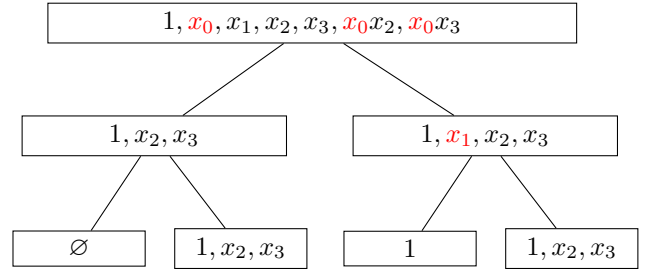
\begin{figure}
		\centering
		\begin{tikzpicture}
	\tikzstyle{code} = [draw,fill=none,minimum size=1.2em]
	\node[code,minimum width=20em] (v2) at (-3,0) {$1,\textcolor{red}{x_0},x_1,x_2,x_3,\textcolor{red}{x_0}x_2,\textcolor{red}{x_0}x_3$};
	\node[code,minimum width=10em] (v1) at (-5.1,-1.5) {$1,x_2,x_3$};
	\node[code,minimum width=10em] (v3) at (-0.7,-1.5) {$1,\textcolor{red}{x_1},x_2,x_3$};
	\node[code,minimum width=5em] (v4) at (-6.1,-3) {$\varnothing$};
	\node[code,minimum width=5em] (v5) at (-3.95,-3) {$1,x_2,x_3$};
	\node[code,minimum width=5em] (v6) at (-1.7,-3) {$1$};
	\node[code,minimum width=5em] (v7) at (0.45,-3) {$1,x_2,x_3$};
	
	\draw  (v2) edge (v1);
	\draw  (v2) edge (v3);
	\draw  (v1) edge (v4);
	\draw  (v1) edge (v5);
	\draw  (v3) edge (v6);
	\draw  (v3) edge (v7);
\end{tikzpicture}
		\caption{Two levels of the code decomposition tree during the standard SC decoding}
		\label{fig:decomp}
	\end{figure}
\end{ex}

Now we are ready to prove Theorem \ref{thm:polaraut}.

\subsection{Proof of Theorem \ref{thm:polaraut}.}
We prove it by contradiction. Assume that the code $\mathcal C^m$ is invariant under the action of $BLTA(\mathbf s, m)$ with a diagonal block of size $s_i$ that is an increasing function of $m$. Then from Theorem \ref{thm:bltasym} it follows that any code $\mathcal C^{(\mathbf j)}\in \phi(\mathcal C^m,i), \mathbf j\in\{-,+\}^{\nu_i}$ is invariant under the permutations on $s_i$ variables $\{x_{\nu_i},\dots,x_{\nu_i+s_i-1}\}$ and consequently the partial derivative w.r.t. any of these variables gives the same code. By definition, the code $\mathcal C^{(-,\mathbf j)}$ is a partial derivative w.r.t. $x_{\nu_i}$ and hence the difference between $R(\mathcal C^{(-,\mathbf j)})$ and $R(\mathcal C^{(\mathbf j)})$ is maximized when the code $\mathcal C^{(\mathbf j)}$ is $s_i$-symmetric. However, from \cite[Proposition 7]{ivanov2021efficiency} we know that this difference actually converges to zero when $s_i$ grows with $m$. More precisely, using the same technique as in \cite[Lemma 7]{reeves2021rm}, we get that $R(\mathcal C^{(\mathbf j)})-R(\mathcal C^{(-,\mathbf j)})\le O(\frac{1}{\sqrt{s_i}})$.

Recall that a polar code by construction is a Plotkin concatenation of two polar codes for the corresponding synthetic channels $W^{(-)}$ and $W^{(+)}$, where $I(W^{(-)})<I(W)$ and $I(W^{(+)})>I(W)$ when $I(W)\notin \{0,1\}$. Applying this definition recursively, we get that code $\mathcal C^{(\mathbf j)}$ needs to be a polar code constructed for the channel $W^{(\mathbf j)}$ and its derivative w.r.t. $x_{\nu_i}$ is a polar code for the channel $W^{(-,\mathbf j)}$. We know from \cite{hassani2014finite} that the ratio of non-perfectly polarized synthetic channels scales as $\Theta(2^{-m/\mu})$, where $\mu$ is the polar scaling exponent, bounded between 3.579 and 6 for any BMS channel. The code $\mathcal C^{(\mathbf j)}$ has length $2^{m-\nu_i}$ and therefore the difference between $R(\mathcal C^{(-,\mathbf j)})$ and $I(W^{(-,\mathbf j)})$ is at most $\Theta(2^{-(m-\nu_i)/\mu})$.

Let us now compare the statements from the two paragraphs. Partial symmetry of code $\mathcal C^{(\mathbf j)}$ gives us an upper bound 
\begin{equation}
	\label{eq:rub}
	R(\mathcal C^{(\mathbf j)})-R(\mathcal C^{(-,\mathbf j)})\le O\left(\frac{1}{\sqrt{s_i}}\right).
\end{equation}
On the other hand, from the nested property of polar codes it follows that
\begin{equation}
	\label{eq:lub}
	R(\mathcal C^{(\mathbf j)})-R(\mathcal C^{(-,\mathbf j)})\ge I(W^{(\mathbf j)})-I(W^{(-,\mathbf j)})-\Theta(2^{-(m-\nu_i)/\mu}).
\end{equation}
We can use again the result from \cite{hassani2014finite} that states that for any interval $[a,b]$, where $a,b\in(0,1)$, the number of channels with capacities that fall into this interval is $\Theta(2^{m(1-1/\mu)})$. Therefore, for any fixed interval $[a,b]$ there exists a channel $W^{(\mathbf j)}$ s.t. $I(W^{(\mathbf j)})\in[a,b]$ and consequently 
$$I(W^{(\mathbf j)})-I(W^{(-,\mathbf j)})\ge \min_{I(\hat W)\in [a,b]}\left(I(\hat W)-I(\hat W^{(-)})\right)>0,$$
where the second inequality holds because all channels in the interval have nontrivial capacities. Hence, we have an upper bound \eqref{eq:rub} that converges to zero because by assumption $s_i$ is an increasing function of $m$ and a lower bound \eqref{eq:lub} that is bounded away from zero, which gives us the contradiction.
	\section{Conclusion}
In this paper, we studied the automorphism group of polar codes. Using our framework of partially symmetric monomial codes, we demonstrated that polar codes can only be invariant under the block-lower-triangular affine permutations if all diagonal blocks are small. More precisely, there cannot exist a sequence of classic polar codes so that the size of the diagonal block grows with the code length. The diagonal blocks induce the affine permutations that are useful for permutation decoding, hence we can conclude that polar codes do not have many useful affine automorphisms.

Our proof shows that optimizing a polar-like code for the SC decoder negatively affects its group of affine automorphisms and can be considered as another confirmation that the design of polar-like codes with good permutation decoding performance is a highly challenging task.

	\enlargethispage{-1.2cm} 
	\IEEEtriggeratref{3}

	\bibliographystyle{IEEEtran}
	\bibliography{biblio}
\end{document}